\documentclass[12pt]{article}

\usepackage{microtype}
\usepackage{times}
\usepackage{amsthm}
\usepackage{amsmath}
\usepackage{amssymb}
\usepackage{mathabx}
\usepackage{amsfonts}
\usepackage{graphicx}
\newtheorem{theorem}{Theorem}[section]

\usepackage[svgnames]{xcolor}
\usepackage{floatpag}
\usepackage{makecell}
\usepackage[raster]{tcolorbox}
\usepackage{microtype}
\usepackage[
  colorlinks=false,
  urlbordercolor=blue,
  pdfborderstyle={/S/U/W 1},
  pdftitle={Netrunner Mate-in-1 or -2 is Weakly NP-Hard},
  pdfauthor={Jeffrey Bosboom and Michael Hoffmann},
  pdflang=en,
]{hyperref}

\usepackage[svgnames]{xcolor}
\usepackage{floatpag}
\usepackage{makecell}
\usepackage[raster]{tcolorbox}

\usepackage{tikz}
\usetikzlibrary{shapes.geometric,positioning}
\usepackage{ifthen}

\usepackage{xparse}
\ExplSyntaxOn
\DeclareExpandableDocumentCommand{\myrepeat}{O{}mm}
 {
  \int_compare:nT { #2 > 0 }
   {
    #3 \prg_replicate:nn { #2 - 1 } { #1#3 }
   }
 }
\ExplSyntaxOff

\newcommand{\ccP}{\textrm{\textsc{P}}}
\newcommand{\ccPSPACE}{\textrm{\textsc{PSPACE}}}
\newcommand{\ccNP}{\textrm{\textsc{NP}}}
\newcommand{\ccEXP}{\textrm{\textsc{EXP}}}
\newcommand{\cccoNP}{\textrm{\textsc{coNP}}}

\newcommand{\threepart}{3-\textsc{partition}}
\newcommand{\twopart}{2-\textsc{partition}}

\title{Mad Science is Provably Hard: Puzzles in Hearthstone's Boomsday Lab are \ccNP-hard}

\author{
  Michael Hoffmann\\
  ETH Zurich\\[-.25em]
  \small{\url{hoffmann@inf.ethz.ch}}
    \and
  Jayson Lynch\\
  University of Waterloo\\[-.25em]
  \small{\url{jayson.lynch@uwaterloo.ca}}
    \and
  Andrew Winslow\\
  UT Rio Grande Valley\\[-.25em]
  \small{\url{andrew.winslow@utrgv.edu }}
}

\begin{document}

\maketitle

\begin{abstract}
We consider the computational complexity of winning this turn  (mate-in-1 or ``finding lethal'') in Hearthstone as well as several other single turn puzzle types introduced in the Boomsday Lab expansion.
We consider three natural generalizations of Hearthstone (in which hand size, board size, and deck size scale) and prove the various puzzle types in each generalization \ccNP-hard. 

\end{abstract}

\section{Introduction}

Collectible card games (CCGs) involve buying some cards, picking a subset or \emph{deck} of those cards, and playing against someone else who separately picked their own deck.
Collectible card games are fun and popular! In 2017 Hearthstone reported over 70 million registered players \cite{HearthstonePlayers}.

With the release of the Hearthstone set, The Boomsday Project, Blizzard introduced a number of puzzles related to the game. All of the puzzles require changing the board to match some desired state in a single turn. The puzzle types are: ``Lethal'', ``Mirror'', ``Board Clear'', and ``Survival''.  %add citation

In Lethal, the player must reduce their opponents health to zero. In Mirror, the player must make both sides of the board exactly the same. This means both boards must have minions of the same type in the same order, and any damage and status effects on those minions must be the same. In Board Clear, the player must ensure there are no more minions on the board, usually by destroying ones already in play. In Survival, the player must return their character to full health.

We will focus on the problem of finding lethal (analogous to ``mate-in-1'' from Chess) since it is the most common puzzle type; however, our proofs can be adapted to show \ccNP-hardness for the other three puzzle types. Even before Blizzard released The Boomsday Project, there were numerous collections of challenging ``lethal puzzles'' available online, e.g. on websites such as \href{http://hearthstonepuzzles.net/}{\color{blue} hearthstonepuzzles.net}, \href{http://www.hsdeck.com/forum/puzzles/}{\color{blue} hsdeck.com}, and \href{http://www.reddit.com/r/HearthPuzzle/}{\color{blue} reddit.com} (also see~\cite{Kotaku}). The frequent difficulty of such lethal puzzles leads to the consideration of the formal computational complexity of such puzzles - can these problems be shown to be computationally intractable under standard complexity-theoretic assumptions?

The Boomsday Project Labs allow for carefully designed situations with non-random decks, cards with reduced costs, complex board states, and even cards that are not normally available in the game. These are all useful tools for designing both puzzles and hardness proofs, however, we are also interested in the problem of finding lethal in which it would occur within a game itself. For this reason, we restrict our proofs to using cards normally available to players and we give a description of how the board state used in the reduction could have been created in a game of Hearthstone. However, for one proof, we do extensively use the puzzle property of the player knowing the contents and order of their deck.

\textbf{Related work.}
The body of work on the computational complexity of games and puzzles has become quite expansive, however, only a small amount of it considers the mate-in-one question. 
Although deciding who will win in a two player game is frequently \ccPSPACE- 
or \ccEXP-complete,\footnote{See~\cite{GPC} for several examples.}, the problem of `mate-in-1' or `finding lethal' is often far less computationally complex because many games have only a polynomial number of moves on any given turn and evaluating whether the new state of the board results in a win is computationally easy. 
Two examples where this problem is interesting are Conway's Phutball for which mate-in-1 is \ccNP-complete~\cite{Phutball} and Checkers for which mate-in-1 is in \ccP~\cite{Checkers}.

Although a fair amount of academic study has gone into collectible card games~\cite{Ward-2009a,zhang2017improving}, less is known about their computational complexity. Magic: The Gathering, perhaps the most well-know CCG, is one example that has been analyzed. 
A recent paper shows that deciding who wins in a game of Magic is undecideable~\cite{MtGTuring, churchill2019magic}. It has been shown that simply deciding if a move is legal is \cccoNP-complete~\cite{chatterjee2016complexity}. Mate-in-1 and Mate-in-2 for another CCG, Android Netrunner, was shown to be weakly \ccNP-hard~\cite{Netrunner}.

\section{Hearthstone}
\label{sec:Hearthstone}
Hearthstone is a popular online CCG made by Blizzard and themed after World of Warcraft. Players are able to purchase virtual cards with which to construct their decks. The game consists of players taking actions on their own turns, including casting spells, summoning and attacking with minions, and controlling heroes with the objective of reducing the enemy hero's health to zero.

New cards come out regularly in sets. Only the base set and most recent sets are allowed in the Standard format. However, all cards are allowed in the Wild format. Blizzard also occasionally changes cards to adjust game-play. Configurations of Hearthstone game considered here take place in the Wild format at the release of The Boomsday Project.\footnote{For further details of the rules of Hearthstone, see \url{http://www.hearthstone.com}.} There are also Solo Adventures which can have unique rules and cards. The Boomsday Project Lab is an example of a Solo Adventure which includes custom rules and cards specifically to facilitate mate-in-1 like puzzles for Hearthstone.

\subsection{Rules Overview}
Here we present a very brief summary of some of the basic rules in Hearthstone which are relevant to the proofs. The cards themselves often have text which specifies additional abilities and rules that may make game-play differ from the typical behavior described in this section.

In Hearthstone, players use \emph{mana} to pay for cards and abilities. Each turn they gain a \emph{mana crystal}, up to 10, and then gain mana equal to their mana crystals. Players also have a \emph{hero} which has \emph{health} and a \emph{hero power}. If a player's hero is ever reduced to zero or less health, that player loses. Hero powers are abilities that cost 2 mana and can be used once a turn.

There are four main types of cards which can be played: \emph{minions}, \emph{spells}, \emph{weapons}, and \emph{traps}. 

\textbf{Spells} typically have a one time effect, such as healing, drawing cards, or doing damage, and are discarded after they are played. Minions are played onto the \emph{Battlefield}. Each player cannot have more than 7 minion on the battlefield at a time. 

\textbf{Minons} have \emph{attack} and \emph{health} which are both non-negative integers. If a minion's health is reduced to zero or lower, it is removed from the battlefield. Minions and Heroes can attack once a turn if they have a positive attack value. When a minion (or hero) attacks, the player chooses an opponent's minion or hero as the target. If the opponent has a minion with \emph{taunt} then a minion with taunt must be selected as the target for any attacks. The attacking and attacked cards simultaneously deal damage to each other equal to their attack value. When a card takes damage its health is reduced by that amount. Minions can also have \emph{abilities} which effect game-play while they are on the battlefield. They can also have \emph{battlecry} or \emph{deathrattle} which triggers an effect (similar to a spell) when the minion is played or dies respectively. Minons without card text (abilities, battlecry, deathrattle) are called \emph{vanilla} minions. For example, \hyperref[hscard:raptor]{Bloodfen Raptor} and \hyperref[hscard:pitfighter]{Pit Fighter} are vanilla minions. %\hyperref[hscard:ogre]{Boulderfist Ogre}. 

\textbf{Weapons} give a hero an attack value. They also have a \emph{durability} which is reduced by 1 every time that hero attacks. If a weapon reaches zero durability it is destroyed. If a player plays a weapon while they already have one in play, the original weapon is destroyed and replaced by a new one.

\textbf{Traps} are not used in any of these proofs as their effects trigger on the opponents turn in response to some action taken by the opponent.

Not all cards can be included in a deck, most notably each hero has a class and decks can only contain \emph{neutral} cards or cards from their class. Decks must normally contain exactly 30 cards, no more than 1 copy of any legendary card and no more than 2 copies of any other card. However, during the game cards may add or remove cards from players' decks and the above constraints only apply while building decks, not in the middle of a game. Games generally begin with cards in a deck being randomly shuffled. However, in The Boomsday Lab puzzles, players might have pre-determined decks with a specific card ordering.

\subsection{Generalizing the game}

In the game of Hearthstone as available to players, the time and complexity of games are limited in several ways.
For instance, each turn has a time limit of 75 seconds (plus animations), games are limited to 89 turns, decks are limited to 60 cards, hands to 10 cards, and boards to 14 minions. The Boomsday Lab puzzles do not have a time limit.
In comparison, the computational complexity of problems is considered as the problem size grows to infinity.

In order to formalize finding lethal into a problem that can be analyzed, we consider generalizing the game in one of several ways, enabling puzzles of arbitrarily large size.
For Hearthstone, we see three natural generalizations: arbitrarily large boards, hands, and decks.
The game configurations obey all rules of Hearthstone, except that turns may take arbitrarily long, they may use arbitrarily many turns (and thus played cards and card copies) to reach, and will be either:
\begin{itemize}
\item \emph{Board-scaled}: the board may have arbitrarily many minions (beyond the~7 permitted in the game).
\item \emph{Hand-scaled}: players' hands may have arbitrarily many cards (beyond the~10 permitted in the game).
\item \emph{Deck-scaled}: players' decks may have arbitrarily many cards (beyond the~60 permitted in the game).%\footnote{Decks may also contain an arbitrary number of copies of any card beyond the 2-card (1-card for Legendary) limit.}
\end{itemize}
%Additionally, the in-game turn time limit of 75 seconds (plus animations) is ignored, allowing players to perform arbitrarily many actions per turn.
Unless otherwise stated, the configurations all occur at turns with~10 mana.
The \emph{lethal problem} of a configuration of Hearthstone is as follows: can the current player reduce their opponents health to zero this turn?

For the Deck-scaled versions of the game we make a further alteration: each player knows the entire content of her deck, including card ordering (i.e. each player has \emph{perfect information} about her deck). 
Normally cards are drawn uniformly at random from those in the deck. 
It would be very interesting to know whether the Deck-scaled version of the game remains hard (or is harder) with random draw.

\subsection{A Preliminary Combo}
\label{sec:prelimcombo}

The reductions below involve large numbers of specific cards.
Although the Boomsday Lab puzzles allow us the freedom to design the precise board state, we show a way of generating cards in a game.
Here we describe a sequence of players yielding arbitrarily many of desired sets of cards (with the caveat of working with very low probability). In particular, many of the cards generate random cards of a given type. In these cases, we assume they happen to generate exactly the cards we desire.

\textbf{The setup.}
On the prior turn, the opponent plays a \hyperref[hscard:millhouse]{Millhouse Manastorm}, causing all spells cast this turn to cost~0 mana.
Our board contains a \hyperref[hscard:brann]{Brann Bronzebeard}, which causes our battlecries to trigger twice, and another vanilla minion. 
Playing \hyperref[hscard:cabalists]{Cabalist's Tome} adds three random Mage spells to your hand. In this case, we assume it generates three copies of \hyperref[hscard:unstable]{Unstable Portal} which adds a random minion to your hand and reduces its mana cost by 3.
Playing two of the three \hyperref[hscard:unstable]{Unstable Portal}s generates a \hyperref[hscard:spellslinger]{Spellslinger} and \hyperref[hscard:void]{Void Terror}. 

\textbf{A cyclic play sequence.}
Since \hyperref[hscard:spellslinger]{Spellslinger} and \hyperref[hscard:void]{Void Terror} cards have mana cost~3, when obtained from \hyperref[hscard:unstable]{Unstable Portal} they have cost~0. \hyperref[hscard:spellslinger]{Spellslinger} has a battlecry which adds a random spell to each player's hand.
Playing \hyperref[hscard:spellslinger]{Spellslinger} yields a \hyperref[hscard:cabalists]{Cabalist's Tome} and a second spell, due to \hyperref[hscard:brann]{Brann Bronzebeard} causing its battlecry to trigger twice. \hyperref[hscard:void]{Void Terror} has a battlecry which destroys adjacent minions.
Playing \hyperref[hscard:void]{Void Terror} between the played \hyperref[hscard:spellslinger]{Spellslinger} and vanilla minion destroys them, recovering space on the board. 

In our hand is now an arbitrary spell (generated from \hyperref[hscard:spellslinger]{Spellslinger}), an arbitrary minion (generated from \hyperref[hscard:unstable]{Unstable Portal}), and a new copy of \hyperref[hscard:cabalists]{Cabalist's Tome}.
No mana has been spent, so this process can be repeated, obtaining a new arbitrary spell and a minion at each iteration.

\textbf{Playing cards.}
If we need to actually play these cards, we can generate \hyperref[hscard:innervate]{Innervates} to gain mana and alternate between generating \hyperref[hscard:void]{Void Terrors} and some other minion that costs 3 mana or less.
This combo requires 4~free minion slots (out of a maximum of~7) and 4~free hand slots, leaving 3 board slots and 6 hand slots for other aspects of later constructions.

\textbf{Obtaining the setup.} 
If neither hero is a mage, we could have generated the initial \hyperref[hscard:cabalists]{Cabalist's Tome} by playing \hyperref[hscard:yogg]{Yogg-Saron, Hope's End}, after having cast at least one spell during the game, and having it cast an \hyperref[hscard:unstable]{Unstable Portal} that generated a \hyperref[hscard:spellslinger]{Spellslinger}.

\section{NP-hardness for Hearthstone}
Here we give three different proofs for the three different generalized versions of Hearthstone. A large battlefield size is considered in Section~\ref{sec:board}, a large hand size is considered in Section~\ref{sec:hand}, and a large deck size is considered in Section~\ref{sec:deck}. The reductions are from \threepart{} and \twopart{} and use similar ideas. In all cases the opponent will have minions with taunt which we must destroy to be able to attack the opponent's hero. These will encode target sums and the attack values of our minions will encode our set of numbers which we want to partition. At this high level the reductions are very simple; the majority of the complication comes from properly constructing the needed attack and health values with cards in the game and a limited hand/deck/board space. It is also important to note that because some cards apply multiplicative factors to attack and health our \twopart{} reductions actually also yield strong NP-hardness for finding lethal in Hearthstone.

\subsection{Hardness of Board-Scaled Lethal}
\label{sec:board}

\begin{theorem}
The lethal problem for board-scaled instances of Hearthstone is \ccNP-hard.
\end{theorem}

\begin{proof}
The reduction is from \threepart.
Let $A = \{a_1, a_2, \dots, a_{3n}\}$ be the input multiset of positive integers that sum to $S$.
The goal is to partition $A$ into $n$ parts that each sum to $S/n$.
The game state is as follows, and is described from the first-person perspective of the current player's turn.

\textbf{Your hero, hand, deck, and board.}
Your hero is Anduin (Priest) with 1 health and \hyperref[hscard:lightsjustice]{Light's Justice} equipped (from \hyperref[hscard:bling]{Blingtron 3000}). \hyperref[hscard:lightsjustice]{Light's Justice} is a weapon allowing the hero to attack for 1 damage.	
Your hand and deck are empty.
Your board consists of $3n$ vanilla minions with attack values $4a_1, 4a_2, \dots, 4a_{3n}$ and 3 health each. 

\textbf{The opponent's hero, hand, deck, and board.}
The opponent's hero is Valeera Sanguinar (Rogue) with~1 health.
The opponent's board consists of $n$ minions, each with taunt, 5~attack, $4S/n$ health, and no other special text.

\textbf{Lethal strategies.}
To win, you must kill all $n$ of the opponent's minions, then do~$1$ damage to the opponent.
Attacking any minion with your weapon causes you to die.
So the opponent's minions must be killed with your minions.
The total attack of your minions is exactly $4S$, so to kill all enemy minions, each minion must not overdamage, i.e. each minion must do exactly its attack damage to some opponent minion.

Thus all of your minions must attack the enemy minions in a way that corresponds exactly to partitioning $a_1, a_2, \dots, a_{3n}$ (your minions' attack values) into $n$ groups of $S/n$ each (your opponent's minions' health values).
Such partitions are exactly the solutions to the \threepart{} instance.

\textbf{Achieving your board state.}
Your board is constructed by using the preliminary combo to generate $3n$ copies of \hyperref[hscard:duskboar]{Duskboar} and $S-n$ copies of \hyperref[hscard:blessing]{Blessing of Kings}. 
\hyperref[hscard:duskboar]{Duskboar}s have $4$~attack and $1$~health and the \hyperref[hscard:blessing]{Blessing of Kings} each add $4$ attack and $4$ health.
These are cast on the \hyperref[hscard:duskboar]{Duskboar}s such that their attacks correspond to the values $4a_1, 4a_2, \ldots, 4a_{3n}$. 
The extra minions needed for the combo are removed later via combo-generated \hyperref[hscard:assassinate]{Assassinate}s.

\textbf{Achieving the opponent's board state.}
The opponent's board can be constructed using the preliminary combo from Section~\ref{sec:prelimcombo}, where the combo generates and the opponent plays $n$ \hyperref[hscard:heckler]{Evil Hecklers} and $S-n$ \hyperref[hscard:blessing]{Blessing of Kings}. 
\hyperref[hscard:heckler]{Evil Hecklers} have 5~attack, 4~health, and taunt. 
\hyperref[hscard:blessing]{Blessing of Kings} gives a minion $+4$ attack and $+4$ health. 
These can be distributed to construct the opponents board which contains $n$ minions with taunt, each with $4S/n+1$ attack, and $4S/n$ health. 
\end{proof}

\subsection{Hardness of Hand-Scaled Lethal}
\label{sec:hand}

\begin{theorem}
The lethal problem for hand-scaled instances of Hearthstone is weakly \ccNP-hard.
\end{theorem}

\begin{proof}
We reduce from \twopart. 
Let $A = \{a_1, a_2, \dots, a_n\}$ be the input set of (exponentially large in $n$) integers that sum to $S$. 
The goal is to partition the integers into two sets that each sum to $S/2$.

\textbf{Your hero, hand, deck, and board.}
Your hero is Jaina (Mage) with 1 health and a \hyperref[hscard:lightsjustice]{Light's Justice} (created by an earlier \hyperref[hscard:bling]{Blingtron 3000}).
Your hand consists of $n$ copies of \hyperref[hscard:bolvar]{Bolvar Fordragon} with attack equal to $b_i = 4a_i - 2$, $n$ copies of \hyperref[hscard:charge]{Charge}, $6n$ copies of \hyperref[hscard:innervate]{Innervate} (to pay for the \hyperref[hscard:bolvar]{Bolvar Fordragon} and \hyperref[hscard:charge]{Charge}).
Your board is empty. \hyperref[hscard:charge]{Charge} allows minions to attack the turn they come into play and Innervate generates additional mana. 

\textbf{Your opponent's hero and board.}
Your opponent's hero is Uther (Paladin) with 2 health.
Your opponents board consists of 2 vanilla minions with taunt, at least~7 attack, $4S$~health (buffed via \hyperref[hscard:blessing]{Blessing of Kings} and \hyperref[hscard:champion]{Blessed Champion}).

\textbf{Lethal strategies.}
In order to win this turn (or at all), you must kill both large minions of the opponent, then do 2~total damage by attacking and using your hero power. 
To win, we must play the \hyperref[hscard:bolvar]{Bolvar Fordragon}s, give them Charge, and attack the minions with taunt such that they both die. 
Since the total attack of all of the \hyperref[hscard:bolvar]{Bolvar Fordragon}s and \hyperref[hscard:charge]{Charge}s equals the health of the two minions, we cannot succeed unless we cast \hyperref[hscard:charge]{Charge} on every \hyperref[hscard:bolvar]{Bolvar Fordragon} exactly once. 
Thus we must allocate the  \hyperref[hscard:bolvar]{Bolvar Fordragon}s between the two enemy minions such that their attack adds up exactly to that of the health of the minions. 

This partition is exactly the solution to the given \twopart{} instance.
The scaling of the attack and health is to deal with the bonus to attack given by \hyperref[hscard:charge]{Charge} and the possibility of $a_i=1$.

\textbf{Achieving your board state.}
In general we will use the infinite combo given earlier in this section to obtain the necessary cards. 
First, we generate all of the cards needed except for the Bolvar Fordragons, as well as some additional cards specified shortly. 
We need to create the \hyperref[hscard:bolvar]{Bolvar Fordragon}s and between them cause many minions to die until we reach the correct values for our set being partitioned. 

Assume the $b_i$ are ordered from largest to smallest. 
We will fill our hand with $n$ copies of \hyperref[hscard:unstable]{Unstable Portal}, $b_1$ copies of \hyperref[hscard:stonetusk]{Stonetusk Boar}, $b_1/2$ copies of \hyperref[hscard:innervate]{Innervate}, and~1 \hyperref[hscard:bestialwrath]{Bestial Wrath}.  
The \hyperref[hscard:bestialwrath]{Bestial Wrath} is cast on an opponent's beast, say \hyperref[hscard:raptor]{Bloodfen Raptor}, to ensure we have a way of killing all of the \hyperref[hscard:stonetusk]{Stonetusk Boar}s.
We play an \hyperref[hscard:unstable]{Unstable Portal} which summons a \hyperref[hscard:bolvar]{Bolvar Fordragon}s. 
We then play a copy of \hyperref[hscard:stonetusk]{Stonetusk Boar} and attack the opponent's \hyperref[hscard:raptor]{Bloodfen Raptor} $b_1-b_2$ times. 
This causes the \hyperref[hscard:bolvar]{Bolvar Fordragons} to gain $b_1-b_2$ attack. 

We then cast another  \hyperref[hscard:unstable]{Unstable Portal} to obtain and play another \hyperref[hscard:bolvar]{Bolvar Fordragons}, then cast $b_2-b_3$ copies of \hyperref[hscard:stonetusk]{Stonetusk Boar}, and attack with them. 
Both \hyperref[hscard:bolvar]{Bolvar Fordragon}s gain $b_2-b_3$ attack from the minions that die. 
We repeat this pattern until we have the $n$th copy of Bolvar Fordragon and we attack with $b_n-1$ copies of Stonetusk Boar. 
Since Bolvar Fordragon starts with~1 attack, we've now caused them to gain the exact amount of attack to have one equal to each of our \threepart{} values.
\end{proof}

\subsection{Hardness of Deck-Scaled Lethal}
\label{sec:deck}

\begin{theorem}
The lethal problem for deck-scaled instances of Hearthstone is \ccNP-hard.
\end{theorem}

\begin{proof}
We reduce from \twopart.
Let $A = \{a_1, a_2, \dots, a_n\}$ be the input set of (exponentially large in $n$) integers that sum to $S$. 
The goal is to partition the integers into two sets that each sum to $S/2$.

\textbf{Your hero, hand, and board.}
Your champion is Uther (Paladin).
Your hand consists of~8 \hyperref[hscard:pitfighter]{Pit Fighter}s and~1 \hyperref[hscard:vigil]{Solemn Vigil}.
Your board consists of~4 frozen vanilla minions. %two more come from the setup sequence

\textbf{Your opponent's hero, hand, deck, and board.}
Your opponent's hero is Uther (Paladin).
Your opponent's board consists of~2 \hyperref[hscard:dummy]{Target Dummy}s buffed to $S/2$ health (via \hyperref[hscard:blessing]{Blessing of Kings} and \hyperref[hscard:champion]{Blessed Champion} to get a large attack and then swapping attack and health with \hyperref[hscard:crazed]{Crazed Alchemist}) and \hyperref[hscard:millhouse]{Millhouse Manastorm}.
Your opponent played \hyperref[hscard:millhouse]{Millhouse Manastorm} last turn, so all spells you cast this turn are free.
Your opponent's deck consists of~1 \hyperref[hscard:bluegill]{Bluegill Warrior}.

\textbf{Your deck.}
Your deck contains the following cards: \hyperref[hscard:vigil]{Solemn Vigil} (SV), \hyperref[hscard:blessing]{Blessing of Kings} (BoK), \hyperref[hscard:champion]{Blessed Champion} (BC), \hyperref[hscard:pitfighter]{Pit Fighter} (PF), and \hyperref[hscard:anyfin]{Anyfin Can Happen} (ACH). 
A \hyperref[hscard:bluegill]{Bluegill Warrior} (BW) had been played earlier and died. 
No other murlocs have died in this game, thus ACH will always summon BW.
The sequence of the first~4 cards, called the \emph{setup sequence}, is SV, PF, SV, PF.

For an integer $n$ ($= b_1 b_2 \ldots b_k$ in binary) with $b_{k-1} b_k = 00$, define the \emph{encoding sequence} of $n$ as a polynomial length sequence of left \emph{bit shifts} (equivalent to multiplying by~2) and \emph{increments} by~8 (equivalent to incrementing by~100 in binary) to obtain $n$. We will be interested in using this to encode the input sequence to the \twopart{} instance.

Define the \emph{integer card sequence} of $n$ to be the sequence of cards obtained by replacing each bit shift and increment in an encoding sequence by BC and BoK, respectively, appending ACH and BC to the beginning of the sequence, and then replacing each card with a SV followed by the card. 
For example, the following sequence of cards encodes $1110100100$: [SV, ACH, SV, BC, SV, BC, SV, BoK, SV, BC, SV, BoK, SV, BC, SV, BC, SV, BoK, SV, BC, SV, BC, SV, BC, SV, BoK].
The complete deck consists of the setup sequence, followed by the integer card sequence for each $a_i$, followed by an ACH.

\textbf{Lethal strategies.}
In order to win this turn, both \hyperref[hscard:dummy]{Target Dummy}s must be killed and~1 damage dealt to the opponent.
To have any possibility of lethal, damage must be done by obtaining ACH via draw.
Thus both \hyperref[hscard:pitfighter]{Pit Fighter}s must be drawn and played, leaving only one open slot to play minions.

Since you've spent all 10 mana, nothing else with positive mana cost (namely the other \hyperref[hscard:pitfighter]{Pit Fighter}s in your hand) can be played.
Thus your hand has between~8 and~10 cards for the remainder of the turn.
Moreover, the interleaving of SV with other spells implies that any attempt to play SV with a hand of~9 cards causes one of the two drawn cards, namely another SV, to be burnt, preventing further card draw. 

At the end of each integer card sequence, to continue to draw into the deck without burning an ACH, we need a minion to die and trigger \hyperref[hscard:cultmaster]{Cult Master} to draw an additional card.
Since \hyperref[hscard:bluegill]{Bluegill Warrior}s are the only minions that can attack, each must be killed at the end of the integer card sequence in which it was drawn and thus only the buffs from one integer card sequence can be applied to a given \hyperref[hscard:bluegill]{Bluegill Warrior}.

Since these buffs on the BW yield a total attack value of $S$, if any buff is not played on a BW then the BWs will have a total attack less than $S$ and thus the two Target Dummies cannot be killed.
Thus any lethal turn involves each BW being buffed with exactly the buffs in the integer card sequence they belong to, and attacking one of the two large minions.

So any lethal play sequence consists of buffing BW to attack values corresponding to $a_1, a_2, \dots, a_n$ and attacking each into one of two opponent \hyperref[hscard:dummy]{Target Dummy}s, followed by a final ACH being drawn and attacking the opponent with the final BW.
Since buffed \hyperref[hscard:bluegill]{Bluegill Warrior}s have $S$ total attack damage, killing both \hyperref[hscard:dummy]{Target Dummy}s requires partitioning the attack values of the buffed BW into two subsets, each of $S/2$ total attack value.

\textbf{Strong hardness.}
Note that unlike the prior reduction from \twopart, this reduction establishes strong NP-hardness. 
This is due to the encoding of exponentially large numbers as minion attack and health values in a polynomial number of cards, one of which (\hyperref[hscard:champion]{Blessed Champion}) doubles minion attack.
because the input to the problem is given by a polynomial number of cards and prior plays in the game. The exponentially large minion sizes are efficiently encoded by a series of plays of BC and BoK.
\end{proof}

\subsection{Adapting to Other Puzzle Types}
In addition to finding ``Lethal'', our proofs can be adapted to show the other three puzzle types introduced in the Boomsday Lab are also \ccNP-hard. In general these proofs require us to carefully construct appropriate minions to remove powerful minions with taunt and allow us to attack the enemy hero. We will show that there are minions which could be chosen that we can attack instead of the enemy hero to accomplish the other goals. 

\paragraph{Survival.} In this puzzle, the player's objective is to restore their hero to full, normally 30, health. To adapt to this case, in each reduction we give our opponent a \hyperref[hscard:mistress]{Mistress of Mixtures} which has two attack, two health, and upon dying restores 4 health to each hero. We make sure the \hyperref[hscard:mistress]{Mistress of Mixtures} only has one health remaining, perhaps by previously damaging it with a \hyperref[hscard:stonetusk]{Stonetusk Boar}. We have your character's health set to 28, so even if you must attack the \hyperref[hscard:mistress]{Mistress of Mixtures} to kill it you will take 2 damage but then regain 4 health returning you to full.

\paragraph{Board Clear.} In this puzzle, the player's objective is to kill all minions on the board. This is achieved in our board scaled reduction. In the other two reductions either your opponent has a leftover \hyperref[hscard:mistress]{Bloodfen Raptor} or you have leftover frozen minions (which we will assume are also Bloodfen Raptors for the sake of simplicity). To fix this issue, we give your opponent an \hyperref[hscard:sheep]{Explosive Sheep} which is a 1 attack, 1 health minion that does 2 damage to all other minions when it dies. Instead of attacking your opponent once you've gotten rid of the taunt minions in the way, attack the \hyperref[hscard:sheep]{Explosive Sheep} whose damage will kill off the remaining unwanted Bloodfen Raptors. In this case we will set your hero's health to 2, so it is greater than the damage done by attacking the \hyperref[hscard:sheep]{Explosive Sheep}.

\paragraph{Mirror.} In this puzzle the player must make both sides of the board identical. We note that if the player clears the board, then both sides will be identically empty, fulfilling the technical requirement if not the spirit of the puzzle. We use the same augmentation as we did with the board clear goal and note that the player has no way of playing the same minions as their opponent and thus cannot fulfill the mirror requirement if any of their opponent's minions are on the board. Thus the only solution is one that involves clearing the board.

\section{Open Problems}
Given the ability to generate an arbitrary number of Hearthstone cards on a single turn, it is not clear Hearthstone puzzles are in \ccNP, or even \ccPSPACE. Obtaining upper bounds on the complexity is a clear open question. Also, for these puzzles, we assume perfect information. If we exploit imperfect information and randomness, even with a bounded number of plays we might suspect the problem is \ccPSPACE-hard.

Hearthstone puzzles also occur on a single turn, thus eliminating the 2-player aspect of the game. What is the complexity of deciding if a player in a game of Hearthstone has a forced win?

Although Magic and Hearthstone are likely the two most famous CCG's at the moment, there have been a number of other such games in the past. It would be interesting to see other examples studied, as well as a general framework for understanding when such games are computationally intractable.

Finally, we have yet to see a formalization of the problem of deck construction and the meta-game involved in most competitive CCGs. Since deck building is such an important and integral part of many of these games, it would be interesting to have a more formal understanding of the questions and process involved.

\subparagraph*{Acknowledgments}

We wish to thank Jeffrey Bosboom for significant feedback and discussion about this paper, as well as LaTeX expertise. We would also like to thank the other participants and especially the organizers (Erik Demaine and Godfried Toussaint) of the Bellairs Research Institute Winter Workshop on Computational Geometry 2015.

\bibliographystyle{plainurl} % LIPIcs
\bibliography{cardgames}

\appendix

\newpage\section{Hearthstone Card Details}\label{sec:hearthstone-cards}

Below is a list of cards used in the constructions and their details. The name of the card is the top line of text and is underlined. Below that is a colored oval denoting the rarity of the card followed by any ability and rules text of the card. The mana cost is given in the upper left hand corner. The attack of the card (for minions and weapons) is given in the lower left. The number in the lower right is the durability for weapons or the health of a minion.

\colorlet{Common}{White}
\colorlet{Rare}{Blue}
\colorlet{Epic}{Purple}
\colorlet{Legendary}{Orange}

\colorlet{Mage}{Blue}
\colorlet{Paladin}{Yellow}
\colorlet{Druid}{Brown}
\colorlet{Rogue}{DarkGrey}
\colorlet{Warrior}{Red}
\colorlet{Warlock}{Purple}
\colorlet{Hunter}{Green}
\colorlet{Neutral}{Grey}

% Name, mana, attack, health, rarity, text, class, db
\newcommand{\hscardboxm}[8]{
\begin{tcolorbox}[width=4.3cm, height=6.02cm, colframe=black, coltitle=black, colback=#7!10, halign=center, lower separated=false, valign lower=bottom, left=0pt, right=0pt]
\vspace{-5mm}
\begin{flushleft}
\hspace{-5mm}
\begin{tikzpicture}
\node[fill=blue!20, draw=black!100, thick, text width=1cm, inner sep=-1mm, align=center, regular polygon, regular polygon sides=6, rounded corners=1mm]{{\fontsize{30}{35}\selectfont \textbf{#2}}};
\end{tikzpicture}
\end{flushleft}

\vspace{0.3cm}

\href{https://www.hearthpwn.com/cards/#8}{{\small \textbf{#1}}} 

\vspace{0.1cm}

\ifthenelse{\equal{#5}{0}}{
\begin{tikzpicture}
    \draw [fill=#7!10, draw=#7!10] (0,0) ellipse (0.15cm and 0.2cm);
\end{tikzpicture}
}{
\begin{tikzpicture}
    \draw [fill=#5] (0,0) ellipse (0.15cm and 0.2cm);
\end{tikzpicture}
}

%\begin{tikzpicture}
%    \draw [fill=#5] (0,0) ellipse (0.15cm and 0.2cm);
%\end{tikzpicture}

\vspace{0.1cm}

{\small {\fontfamily{iwona}\selectfont #6}}

\tcblower

\begin{tikzpicture}[remember picture, overlay]
\node[xshift=2mm, fill=yellow!40, inner sep=1pt, circle, thick, draw=black!100] (a) {{\fontsize{30}{35}\selectfont \textbf{#3}}};
\node[fill=red!40, inner sep=1pt, circle, thick, draw=black!100, right=2.57cm of a] (b) {{\fontsize{30}{35}\selectfont \textbf{#4}}};
\end{tikzpicture}

\end{tcolorbox}
}

% Name, mana, rarity, text, class, db
\newcommand{\hscardboxs}[6]{
\begin{tcolorbox}[width=4.3cm, height=6.02cm, colframe=black, coltitle=black, colback=#5!10, halign=center, halign lower=flush right, lower separated=false, left=0pt, right=0pt]
\vspace{-5mm}
\begin{flushleft}
\hspace{-5mm}
\begin{tikzpicture}
\node[fill=blue!20, draw=black!100, thick, text width=1cm, inner sep=-1mm, align=center, regular polygon, regular polygon sides=6, rounded corners=1mm]{{\fontsize{30}{35}\selectfont \textbf{#2}}};
\end{tikzpicture}
\end{flushleft}
\vspace{0.3cm}
\href{https://www.hearthpwn.com/cards/#6}{{\small \textbf{#1}}} \\
\vspace{0.1cm}

\ifthenelse{\equal{#3}{0}}{
\begin{tikzpicture}
    \draw [fill=#5!10, draw=#5!10] (0,0) ellipse (0.15cm and 0.2cm);
\end{tikzpicture}
}{
\begin{tikzpicture}
    \draw [fill=#3] (0,0) ellipse (0.15cm and 0.2cm);
\end{tikzpicture}
}
\vspace{0.1cm}
\\
{\small {\fontfamily{iwona}\selectfont #4}}
\tcblower
\end{tcolorbox}
}

%Cards used in Hearthstone Proofs
%Millhouse Manastorm
%Cabalst's Tome
%Spellslinger
%Brann Bronzebeard
%Yogg-Soggoth
%Blingtron 3000
%Evil Hecklers
%Blessing of Kings
%Duskboar
%Cold Blood
%Bolvar Fordragon
%Charge
%Invigorate
%Blessed Champion
%Bloodfen Raptor
%Beastial Wrath
%Solemn Vigil
%Bluegill Warrior
%Anyfin Can Happen
%Pit Fighter
%Boulderfist Ogres
%Cult Master
%Target Dummies
%Void Terror
%Stonetusk Boar
%Unstable Portal
%Light's Justice

%\begin{tcbraster}[raster columns=3, raster rows=3, raster row skip=1cm, raster column skip=1cm] 

\vspace{1.5cm}

\begin{tabular}{c c c}
\label{hscard:anyfin}
 \hscardboxs{Anyfin Can Happen}{10}{Rare}{Summon 7 Murlocs\\that died this game.}{Paladin}{27240-anyfin-can-happen}&

\label{hscard:assassinate}
\hscardboxs{Assassinate}{5}{0}{Destroy an enemy\\minion.}{Rogue}{568-assassinate} &

\label{hscard:bestialwrath}
\hscardboxs{Bestial Wrath}{1}{Epic}{Give a friendly Beast\\+2 Attack and \textbf{Immune}\\this turn.}{Hunter}{304-bestial-wrath} 
\\

\label{hscard:blessing}
\hscardboxs{Blessing of Kings}{4}{0}{Give a minion +4/+4.\\\emph{(+4 Attack/+4 Health)}}{Paladin}{29-blessing-of-kings} & 

\label{hscard:champion}
\hscardboxs{Blessed Champion}{5}{Rare}{Double a minion's\\attack.}{Paladin}{7-blessed-champion} & 

\label{hscard:bling}
\hscardboxm{Blingtron 3000}{5}{3}{4}{Legendary}{\textbf{Battlecry:} Equip a random\\weapon for each player.}{Neutral}{12183-blingtron-3000} 
\\

\end{tabular}
\newpage
\begin{tabular} {c c c}

\label{hscard:raptor}
\hscardboxm{Bloodfen Raptor}{2}{3}{2}{0}{~}{Neutral}{576-bloodfen-raptor}&

\label{hscard:bluegill}
\hscardboxm{Bluegill Warrior}{2}{2}{1}{0}{\textbf{Charge}}{Neutral}{289-bluegill-warrior}&

\label{hscard:bolvar}
\hscardboxm{Bolvar Fordragon}{5}{1}{7}{Legendary}{Whenever a friendly minion\\dies while this is in your\\hand, gain +1\\Attack.}{Paladin}{12244-bolvar-fordragon} 
\\

\label{hscard:ogre}
\hscardboxm{Boulderfist Ogre}{6}{6}{7}{0}{~}{Neutral}{60-boulderfist-ogre}&

\label{hscard:brann}
\hscardboxm{Brann Bronzebeard}{3}{2}{4}{Legendary}{Your~\textbf{Battlecries}~trigger\\twice.}{Neutral}{27214-brann-bronzebeard} &

\label{hscard:cabalists}
\hscardboxs{Cabalist's Tome}{5}{Epic}{Add 3 random Mage\\spells to your hand.}{Mage}{33155-cabalists-tome} 
\\

\label{hscard:charge}
\hscardboxs{Charge}{1}{0}{Give a friendly minion\\\textbf{Charge}. It can't attack\\heroes this turn.}{Warrior}{646-charge} & % Pre or post nerf? Both work since we attack minions

%\label{hscard:coldblood}
%\hscardboxs{Cold Blood}{1}{Common}{Give a minion +2\\Attack. \textbf{Combo:} +4 Attack instead.}{Rogue}{https://www.hearthpwn.com/cards/92-cold-blood} &

\label{hscard:crazed}
\hscardboxm{Crazed Alchemist}{2}{2}{2}{Rare}{\textbf{Battlecry:} Swap the Attack and Health of a minion.}{Neutral}{https://hearthstone.gamepedia.com/Crazed_Alchemist} &

\label{hscard:cultmaster}
\hscardboxm{Cult Master}{4}{4}{2}{Common}{Whenever one of your other\\minions dies, draw a card.}{Neutral}{140-cult-master} 
\\

\end{tabular}
\newpage
\begin{tabular}{c c c}

\label{hscard:duskboar}
\hscardboxm{Duskboar}{2}{4}{1}{Common}{~}{Neutral}{35251-duskboar} &

\label{hscard:heckler}
\hscardboxm{Evil Heckler}{4}{5}{4}{Common}{\textbf{Taunt}}{Neutral}{22390-evil-heckler} &

\label{hscard:sheep}
\hscardboxm{Explosive Sheep}{2}{1}{1}{Common}{\textbf{Deathrattle: } Deal 2 damage to all minions.}{Neutral}{12180-explosive-sheep}
\\

\label{hscard:innervate}
\hscardboxs{Innervate}{0}{0}{Gain 1 Mana Crystal\\this turn only.}{Druid}{77035-innervate} & 

\label{hscard:lightsjustice}
\hscardboxm{Light's Justice}{1}{1}{4}{0}{}{Neutral}{76-stonetusk-boar} &

\label{hscard:millhouse}
\hscardboxm{Millhouse Manastorm}{2}{4}{4}{Legendary}{\textbf{Battlecry:} Enemy spells\\cost (0) next turn.}{Neutral}{339-millhouse-manastorm} 
 \\

\label{hscard:mistress}
\hscardboxm{Mistress of Mixtures}{1}{2}{2}{Common}{\textbf{Deathrattle: } Restore 4 Health to each hero.}{Neutral}{49646-mistress-of-mixtures} &

\label{hscard:pitfighter}
\hscardboxm{Pit Fighter}{5}{5}{6}{Common}{~}{Neutral}{22375-pit-fighter} &

\label{hscard:spellslinger}
\hscardboxm{Spellslinger}{3}{3}{4}{Legendary}{\textbf{Battlecry:} Add a random\\spell to each player's\\hand.}{Mage}{22299-spellslinger}
 \\
 
\end{tabular}

\newpage
\begin{tabular}{c c c}
 
\label{hscard:stonetusk}
\hscardboxm{Stonetusk Boar}{1}{1}{1}{0}{\textbf{Charge}}{Neutral}{76-stonetusk-boar} & 

\label{hscard:dummy}
\hscardboxm{Target Dummy}{0}{0}{2}{Rare}{\textbf{Taunt}}{Neutral}{12288-target-dummy} &
 
\label{hscard:unstable}
\hscardboxs{Unstable Portal}{2}{Rare}{Add a random minion\\to your hand. It costs\\(3) less.}{Mage}{12178-unstable-portal} 
 \\

\label{hscard:vigil}
\hscardboxs{Solemn Vigil}{5}{Common}{Draw 2 cards. Costs (1)\\less for each minion\\that died this turn.}{Paladin}{14453-solemn-vigil} &
\label{hscard:void}
\hscardboxm{Void Terror}{3}{3}{3}{Rare}{\textbf{Battlecry:} Destroy both\\adjacent minions and gain\\their Attack and Health.}{Warlock}{119-void-terror}& 
\label{hscard:yogg}
 \hscardboxm{Yogg-Saron, Hope's End}{10}{7}{5}{Legendary}{{\footnotesize \textbf{Battlecry:}~Cast~a~random~spell for each spell you've cast this\\game \emph{(targets chosen\\randomly)}.}}{Neutral}{33168-yogg-saron-hopes-end} \\

\end{tabular}

\end{document}